\RequirePackage[l2tabu,orthodox]{nag}
\documentclass
[11pt,letterpaper]
{article} 

\usepackage[notes=true,later=false,camera=false]{dtrt}
\usepackage[utf8]{inputenc}
\usepackage{ stmaryrd }
\usepackage{xspace,enumerate}
\usepackage[T1]{fontenc}
\usepackage[full]{textcomp}
\usepackage[american]{babel}
\usepackage{mathtools}

\renewcommand{\hat}[1]{\widehat{#1}}
\usepackage{amsthm}
\usepackage{thmtools}
\usepackage[capitalise,nameinlink]{cleveref}

\usepackage{empheq}

\definecolor{citeblue}{HTML}{0055cc}
\hypersetup{
colorlinks=true,
urlcolor=citeblue,
linkcolor=NavyBlue,
citecolor=PineGreen,
linktocpage=true,
}
\renewcommand*\backref[1]{\ifx#1\relax \else (pg. #1) \fi}

\renewcommand{\tilde}{\widetilde}

\usepackage{paralist}
\usepackage{turnstile}
\usepackage[framemethod=TikZ]{mdframed}
\mdfsetup{frametitlealignment=\center}
\usepackage{tikz}
\usepackage{caption}
\DeclareCaptionType{Algorithm}
\usepackage{newfloat}

\newtheorem{theorem}{Theorem}[section]
\newtheorem{lemma}[theorem]{Lemma}
\newtheorem*{lemma*}{Lemma}

\newtheorem{proposition}[theorem]{Proposition}
\newtheorem{fact}[theorem]{Fact}

\theoremstyle{definition}
\newtheorem{definition}[theorem]{Definition}
\newtheorem*{definition*}{Definition}
\newtheorem{remark}[theorem]{Remark}

\usepackage[linesnumbered,ruled]{algorithm2e}
\crefname{lemma}{Lemma}{Lemmas}
\crefname{fact}{Fact}{Facts}
\crefname{theorem}{Theorem}{Theorems}
\crefname{mtheorem}{Theorem}{Theorems}
\crefname{itheorem}{Theorem}{Theorems}
\crefname{corollary}{Corollary}{Corollaries}
\crefname{claim}{Claim}{Claims}
\crefname{example}{Example}{Examples}
\crefname{algorithm}{Algorithm}{Algorithms}
\crefname{problem}{Problem}{Problems}
\crefname{definition}{Definition}{Definitions}
\crefname{equation}{Eq.}{Eq.}
\crefname{strategy}{Strategy}{Strategies}
\crefname{observation}{Observation}{Observations}

\Crefname{algocf}{Algorithm}{Algorithms}

\usepackage[
letterpaper,
top=1.2in,
bottom=1.2in,
left=1in,
right=1in]{geometry}
\usepackage{newpxtext} 
\usepackage{textcomp} 
\usepackage[scr=rsfso]{mathalfa}
\usepackage{bm} 

\usepackage{microtype}

\usepackage{footnotebackref}

\newcommand{\supp}{\operatorname{supp}}
\newcommand{\AND}{\operatorname{AND}}

\allowdisplaybreaks
\newcommand{\FormatAuthor}[3]{
\begin{tabular}{c}
#1 \\ {\small\texttt{#2}} \\ {\small #3}
\end{tabular}
}

\newcommand{\R}{{\mathbb R}}
\newcommand{\N}{{\mathbb N}}

\newcommand{\eps}{\varepsilon}

\newcommand{\F}{{\mathbb F}}

\newcommand{\1}{\mathbf{1}}

\newcommand{\C}{\mathbb C}

\newcommand{\cG}{\mathcal G}

\newcommand{\poly}{\mathrm{poly}}

\newcommand{\seq}{\subseteq}

\newcommand{\polylog}{\operatorname{polylog}}

\renewcommand{\geq}{\geqslant}

\renewcommand{\leq}{\leqslant}
\renewcommand{\le}{\leqslant}
\renewcommand{\preceq}{\preccurlyeq}

\renewcommand{\epsilon}{\varepsilon}

\newcommand{\ignore}[1]{}

\newcommand{\Cay}{\mathrm{Cay}}
\newcommand{\id}{\operatorname{id}}

\newcommand{\vol}{\operatorname{Vol}}
\newcommand{\Ree}{\operatorname{Re}}
\newcommand{\Imm}{\operatorname{Im}}

\begin{document}
\author{
 \begin{tabular}{cc}
 \FormatAuthor{Arpon Basu}{arpon.basu@princeton.edu}{Princeton University} &
 \FormatAuthor{Pravesh K. Kothari}{kothari@cs.princeton.edu}{Princeton University} \\
  & \\
  \FormatAuthor{Raghu Meka\thanks{Supported by NSF EnCORE: Institute for Emerging CORE Methods in Data Science Award \#2217033 and NSF AF: Small Award \#2425350}}{raghum@cs.ucla.edu}{University of California, Los Angeles} &
  \FormatAuthor{Stefan Tudose}{studose@princeton.edu}{Princeton University} 
  \end{tabular}
  }

\title{Optimal Sparsifiers for Abelian Cayley Graphs}

\maketitle

\abstract{We prove that for every Cayley graph $\mathcal{G}$ over any finite abelian group $G$, there is a weighted Cayley graph with $O(\log |G|)$ generators that is a spectral sparsifier for $\mathcal{G}$. This bound is optimal. Applying our bound to the group $G = \F_2^n$ yields, as a corollary, $O(n/\eps^2)$-sized code sparsifiers for $\F_2$-linear codes, improving on the work of Khanna, Putterman and Sudan \cite{KhannaPS24} who obtained a similar result with an additional $\polylog(n)$ loss. 

Our proof is strongly inspired by a recent work of Reis and Rothvoss \cite{ReisR26} for the construction of $\ell_1$-sparsifiers. Following their work, the abelian Cayley sparsification problem can be reduced to establishing a lower bound for the volume of a certain natural convex body. This volume bound follows from a short, elementary argument that relies on character symmetry. 
}

\pagenumbering{arabic}

\section{Introduction}
Sparsification refers to the process of compressing an object (say a graph, or a code, or a set system) while still retaining some essential features of the object. Sparsification was first introduced by Bencz\'ur and Karger \cite{BenczurK96} in the context of \emph{cut sparsification}, who showed that with $\leq\tilde{O}(n/\eps^2)$ edge weights, one could preserve the values of \emph{all} $2^n$ cuts of a graph up to a multiplicative $1\pm\eps$ factor.

Since then cut sparsification, and generalizations such as spectral sparsification \cite{SpielmanT11,SpielmanS11,BatsonSS14} have proved to be very useful in graph algorithms \cite{BenczurK96,SpielmanS11}, in solving Laplacian linear systems \cite{SpielmanT04}, in reducing the space usage of sublinear time algorithms \cite{AhnGM12b,McGregor14,AbrahamDKKP16,KapralovLMMS17}, and many other applications.

Given the success of sparsification as a paradigm, much effort has also been invested in generalizing graph sparsification to more general objects, such as hypergraphs \cite{KoganK15,ChenKN20,KapralovKTY21,KapralovKTY21a,JambulapatiLS23,Lee23}, codes \cite{KhannaPS24,KhannaPS25,BrakensiekG25}, and CSPs \cite{KoganK15,FiltserK17,ButtiZ20}, to mention a few applications.

One such generalization which has been investigated is the notion of \emph{Cayley sparsification} \cite{KhannaPS24,KhannaPS25,HsiehLMPZ26,BasuKLM26}, which is what we study in this paper. 

We now formally define Cayley graphs and Cayley sparsification:
\begin{definition}[Cayley Graphs]
    Let $G$ be a group, and let $S\subset G$ be a \textbf{symmetric} subset of $G$, i.e. $s\in S$ iff $s^{-1}\in S$. The Cayley graph $\Cay(G, S)$ is a graph on $G$ where $g, g'\in G$ are adjacent if $g^{-1}g'\in S$. In general we also consider weighted Cayley graphs, wherein we have a symmetric weight function $w:S\to\R_{\geq 0}$ (satisfying $w(s) = w(s^{-1})$ for all $s\in S$), and the weighted Cayley graph $\Cay(G, S, w)$ is the graph $\Cay(G, S)$ where the edge $\{g, g'\}\in E(\Cay(G, S))$ receives the weight $w(g^{-1}g')$. Thus unweighted Cayley graphs can be viewed as possessing the weight function $w:S\to\{1\}$. 
\end{definition}
Throughout the paper we will assume that our generating set $S$ is symmetric, and thus our Cayley graphs $\Cay(G, S)$ are always undirected.

We can now define Cayley sparsification:
\begin{definition}[Cayley Sparsification]
    Let $\eps\in(0, 1)$ be some parameter. Given a (weighted) Cayley graph $\mathcal{G}:= \Cay(G, S, w)$, we say $\Cay(G, S', w')$ is an $\eps$-Cayley sparsifier for $\mathcal{G}$ if 
    \[(1 - \eps)L\preceq L'\preceq(1 + \eps)L,\]
    where $L$ (resp. $L'$) refers to the Laplacian of $\Cay(G, S, w)$ (resp. $\Cay(G, S', w')$), and $\preceq$ refers to the Loewner order on the space of Hermitian matrices, i.e. $A\preceq B$ iff $B - A$ is positive semidefinite (PSD).

    If $\Cay(G, S', w')$ is an $\eps$-Cayley sparsifier for $\Cay(G, S, w)$, we write $\Cay(G, S', w')\approx_\eps\Cay(G, S, w)$. 
\end{definition}

Invoking spectral graph sparsification primitives such as \cite{SpielmanS11,BatsonSS14} on $\mathcal{G}$ will produce a sparsifier with the right number of edges. However, they crucially lose the algebraic structure and will not produce a sparsifier which is also a Cayley graph which we focus on.

Before we introduce our main result, we recall the Cayley sparsification results of \cite{KhannaPS24,KhannaPS25,HsiehLMPZ26,BasuKLM26} to serve as a point of comparison for our results:
\begin{theorem}[Cayley sparsification over $\F_2^n$ \cite{KhannaPS24,KhannaPS25}]
\label{thm:kps}
    Let $G = \F_2^n$. For any (weighted) Cayley graph $\Cay(G, S, w)$ there exists an $\eps$-Cayley sparsifier $\Cay(G, S', w')$ with $S'\seq S$ such that 
    \[|S'|\leq O(\eps^{-2}n\polylog(n))\leq O(\eps^{-2}(\log|G|) \cdot \poly(\log \log |G|)).\] 
    Furthermore $(S', w')$ can be computed in randomized $\poly(n, |S|, \eps^{-1}) = \poly(\log|G|, |S|, \eps^{-1})$ time.
\end{theorem}

\begin{theorem}[Cayley sparsification over arbitrary groups \cite{HsiehLMPZ26,BasuKLM26}]
    Let $G$ be an arbitrary group, possibly non-abelian. For any Cayley graph $\Cay(G, S)$ there exists an $\eps$-Cayley sparsifier $\Cay(G, S', w')$ such that 
    \[|S'|\leq O(\eps^{-2}(\log|G|)^{4}).\] 
    Furthermore $(S', w')$ can be computed in randomized $\poly(|G|, |S|, \eps^{-1}) = \poly(|G|, \eps^{-1})$ time.
\end{theorem}

Given these results, it is natural to wonder what the \emph{optimal} possible bound on the size of Cayley sparsifiers is. We settle the optimality question in this paper for \emph{abelian Cayley graphs}:

\begin{theorem}[Optimal Abelian Cayley Sparsification]
\label{thm:cayley-sparsifier}
    Let $G$ be an abelian group and let $\eps\in(0, 1)$. For any (weighted) Cayley graph $\Cay(G, S, w)$ there exists an $\eps$-Cayley sparsifier $\Cay(G, S', w')$ such that 
    \[|S'|\leq O(\eps^{-2}\cdot\log|G|).\] 
    Furthermore $(S', w')$ can be computed in randomized $\poly(|G|, \eps^{-1})$ time.
 \end{theorem}
 \begin{remark}
     We mention a few salient points about the theorem:
     \begin{enumerate}[(1)]
        \item \cite{KhannaPS24} shows that \emph{code sparsification} of $\F_2$-linear codes is equivalent to sparsifying Cayley graphs over $\F_2^n$. Consequently, $n$-dimensional $\F_2$-linear codes admit $O(n/\eps^2)$-sized code sparsifiers, thus removing the log factors in \cref{thm:kps}. 
        \item Note that the sparsified set $S'$ admits a size upper bound \emph{independent} of the initial weight function $w$, i.e. the bound of $O(\eps^{-2}\log|G|)$ holds regardless of what $w$ is.
     \end{enumerate}
\end{remark}

Since \cref{thm:cayley-sparsifier} is about abelian Cayley graphs, throughout the rest of the paper we assume our groups to be abelian, and we use additive notation to represent it, i.e. $G = (G, + , 0)$ refers to an abelian group whose identity is written as $0$.

\parhead{Optimality of \cref{thm:cayley-sparsifier}} The dependence of $\log(N)$ in \cref{thm:cayley-sparsifier} is optimal. This result was essentially present in the work \cite{BasuKLM26}, and we formally record it below:

\begin{theorem}[Optimality of Abelian Cayley Sparsification]
    Let $G$ be an arbitrary abelian group of size $N$. Then for every $\eps\in(0, 1)$, there exists a constant $c(\eps) > 0$ and a subset $S\subset G$ of size $|S|\geq c(\eps)\log N$ such that $\Cay(G, S)$ does not admit any $\eps$-sparsifier of size $< |S|$, i.e. $\Cay(G, S)$ can \textbf{not} be $\eps$-sparsified.
\end{theorem}
\begin{proof}
    The claim follows by combining \cite[Theorem~3.4 and Lemma~4.6]{BasuKLM26}.\footnote{The cited results can be found in the arXiv version \url{arxiv.org/abs/2508.08169}}
\end{proof}

Beyond just the theory of sparsification, Cayley sparsifiers have connections to many other areas of computer science, and as such the optimality of \cref{thm:cayley-sparsifier} leads to some interesting observations in those areas. For instance, a famous theorem of Alon-Roichman \cite{AlonR94} states that the complete graph on $N$ vertices admits an $\eps$-\emph{Cayley sparsifier} of size $O(\eps^{-2}\log N)$ for any group $G$ of size $N$. Equivalently, the Alon-Roichman theorem furnishes $O(\eps^{-2}\log N)$-sized Cayley sparsifiers for $\Cay(G, S)$ where $S = G\setminus\{\id_G\}$, and thus \cref{thm:cayley-sparsifier} generalizes the Alon-Roichman theorem for abelian groups $G$, since \cref{thm:cayley-sparsifier} obtains $O(\eps^{-2}\log N)$-sized Cayley sparsifiers for $\Cay(G, S)$ for \emph{arbitrary} symmetric sets $S\subset G$. 

Another connection along these lines is the following: Unweighted\footnote{here by unweighted we mean that all elements of the generating set have the same weight} Cayley sparsifiers of the complete graph $\Cay(\F_2^n, \F_2^n)$ naturally correspond to $\F_2$-linear $\eps$-biased codes. Consequently, a deterministic construction of a $O(n/\eps^2)$-sized unweighted Cayley sparsifier for $\Cay(\F_2^n, \F_2^n)$ would imply an explicit construction of an optimal $\eps$-biased $\F_2$-linear code! Note that \cref{thm:cayley-sparsifier} obtains the correct size bound of $O(\eps^{-2}n)$, but unfortunately computing sparsifiers via \cref{thm:cayley-sparsifier} seems to take time $\poly(|G|) = 2^{\Omega(n)}$ in general. \cite{KhannaPS25} obtains a $O(\eps^{-2}n\polylog(n))$-sized sparsifier in $\poly(n)$ time, but unfortunately their algorithm is randomized. Obtaining the correct size tradeoff with a deterministic $\poly(n)$ time algorithm thus remains a significant challenge. 

For further connections between Cayley sparsifiers, pseudorandomness and complexity theory, we refer the reader to \cite{JalanM21}.

\section{Preliminaries}

\parhead{Groups and Characters} Let $G = (G, +, 0)$ be an abelian group of size $N$, and let $\C^*$ be the set of all non-zero complex numbers. A map $\chi:G\to\C^*$ is said to be a character of $G$ if $\chi(g + g') = \chi(g)\chi(g')$ for all $g, g'\in G$. We refer to as $\hat{0}$ the \emph{trivial character} corresponding to the map $G\ni g\mapsto 1\in\C^*$. 

An abelian group of size $N$ possesses exactly $N$ characters, whose collection we denote as $\hat{G}$. $\hat{G}$ can be equipped with the structure of an abelian group, where for $\chi, \chi'\in\hat{G}$ we define 
\[\left(\chi +_{\hat{G}} \chi'\right)(g):= \chi(g)\chi'(g).\]
It is easily seen that $\hat{G} = (\hat{G}, +_{\hat{G}}, \hat{0})$ is an abelian group, with the trivial character $\hat{0}$ serving as the group identity. Henceforth we shall treat $\hat{G}$ as an abstract abelian group and write $\pm_{\hat{G}}$ simply as $\pm$.

A set $S\subset G$ is said to be symmetric if $S = -S$. A map $f:G\to\C$ is said to be symmetric if $f(g) = f(-g)$ for all $g\in G$.

\parhead{Graph Laplacians} Let $\cG = \cG(V, E, w)$ be a weighted undirected graph, where $w:E\to\R_{\geq 0}$ is some weight function. We shall view unweighted graphs simply as possessing the weight function $w:E\to\{1\}$. Define the \emph{degree matrix} $D_\cG$ of $\cG$ to be a diagonal matrix in $\R^{V\times V}$ with $D_\cG(v, v):= \deg(v) = \sum_{e: v\in e}w(e)$ being equal to the degree of the vertex $v$ in $\cG$. Also define the adjacency matrix $A_\cG\in\R^{V\times V}$ as $A_\cG(v, v'):= \1(\{v, v'\}\in E(\cG))\cdot w(\{v, v'\})$. Then the \emph{Laplacian} of $\cG$ is defined simply to be $D_\cG - A_\cG$. 

We recall the following standard fact about the Laplacian eigenvalues of \emph{abelian} Cayley graphs $\Cay(G, S, w)$. For any character $\chi\in\hat{G}$ and $g\in G$, define $\lambda(\chi, g):= 1 - \Ree\chi(g)$. 
\begin{fact}[Abelian Cayley Laplacian Eigenvalues]
\label{fact:cayley-eig}
    Let $L\in\R^{G\times G}\subset\C^{G\times G}$ be the Laplacian of a weighted Cayley graph $\Cay(G, S, w)$, where $G$ is an abelian group, and $S\subset G$ is symmetric. The characters in $\hat{G}$ form an orthogonal eigenbasis of $L$, with the eigenvalue corresponding to $\chi\in\hat{G}$ being $\sum_{s\in S}w(s)\lambda(\chi, s)$. 
\end{fact}

For a complex number $z = x + iy$, where $x, y\in\R$ and $i = \sqrt{-1}$, we define $\Ree(z):= x$ and $\Imm(z):= y$. 

A set $K\subset\R^m$ is called a convex body if it is a closed compact convex set. For any $\lambda\in\R$, write $\lambda K:= \{\lambda x:x\in K\}$. $K$ is called centrally symmetric if $K = -K$. For a convex body $K\seq\R^m$, and any $T\seq[m]$, write 
\[K_T:= \{x\in\R^T: \R^T\times\{0\}^{[m]\setminus T}\ni(x, 0)\in K\}.\]

Let $\vol_S$ denote the usual Lebesgue measure on $\R^S$, i.e. $\vol_S([0, 1]^S) = 1$. If $|S| = m$, we also sometimes write $\vol_S$ as $\vol_m$, or simply $\vol$, if $S$ is clear from the context.

For two vectors $v, w\in\R^n$, define $v\odot w\in\R^n$ to be their coordinate-wise product, i.e. $(v\odot w)_i:= v_i\cdot w_i$ for all $i\in[n]$.

We also define the notion of group-invariant random processes, which plays an important role in our proof:
\begin{definition}[group-invariant Random Processes]
    Let $H$ be an abelian group. We say that $W:= (W_h)_{h\in H}$ is a $H$-invariant random process if $(W_h)_{h\in H}$ are random variables such that for any $h_0\in H$, the joint distribution of $(W_{h + h_0})_{h\in H}$ is the same as the joint distribution of $(W_h)_{h\in H}$.
\end{definition}

Finally, we note an elementary fact about probability theory:
\begin{fact}[Rotationally Invariant Measure with Uniform Marginal]
\label{fact:rot-uniform}
    There exists a distribution $\nu = (X, Y)$ supported on $\R^2$ which is rotationally invariant and the distribution of (the marginal random variable) $X$ is the uniform distribution on $[-1, 1]$. For the sake of concreteness, 
    \[\R^2\ni z\mapsto\frac{\1(|z| < 1)}{2\pi\sqrt{1 - |z|^2}} \]
    is the density function of such a distribution. Here $|\cdot|$ stands for the usual $\ell_2$ norm in $\R^2$. 
\end{fact}

\section{Proof Overview}
We next highlight the main ideas behind the proof of \cref{thm:cayley-sparsifier}. 

Fix a weighted Cayley graph $\Cay(G,S,w)$. For brevity, let $N = |G|$. As a first step, following the work of \cite{ReisR26}, the first idea is that instead of building the sparsifier in one shot, we show that as long as $|S| \gg \log N/\varepsilon^2$, we can shrink the size of $|S|$ by a constant factor while incurring low-error. We have to do some bookkeeping to make sure the errors don't add up; this can be done by choosing the parameters appropriately. For now, let us focus on a single step of this shrinkage. 

To this end, for a desired error $\epsilon$, let us define the \emph{sparsification polytope}:   
\[
\begin{aligned}\label{eq:sparspoly}
Q_{\eps}:=\Bigl\{X\in[-1,1]^S:\;&
(1-\eps)\sum_{s\in S}w(s)\lambda(\chi,s)
\le
\sum_{s\in S}(1+X_s)w(s)\lambda(\chi,s)\\
&
\le
(1+\eps)\sum_{s\in S}w(s)\lambda(\chi,s)
,\quad\forall\,\chi\in\hat{G}
\Bigr\}.
\end{aligned}
\]

That is $Q_\eps$ is the set of all vectors $X \in [-1,1]^S$ such that 
for the new weight function $w_X:S \rightarrow \R$ defined by $w_X(s) = w(s) (1 + X_s)$, we have $\Cay(G,S,w_X) \approx_\eps \Cay(G,S,w)$.

Observe that $Q_\eps$ is clearly a polytope and it is also easy to see that it is symmetric. The key step is to show that as long as $|S|$ is sufficiently large, there exists a vector $X \in \alpha Q_\eps \cap [-1,1]^S$ such that many of the coordinates of $X$ are $-1$ (i.e., $X$ is a good \emph{partial coloring}). If this is true, then replacing $w$ by $w_X$ shrinks the support of $S$ significantly. To do so, we use a volume argument: If the volume of $Q_\eps$ is sufficiently large, then (essentially by Minkowski's theorem), it should contain many lattice points and hence also many points with one of the coordinates being $-1$. Such volume arguments have been used in discrepancy literature before (e.g., \cite{Rothvoss17}). 

We use the following concrete form of this argument from \cite{ReisR26}:
\begin{theorem}[Theorem 7 in \cite{ReisR26}]
\label{thm:good-partial-coloring}
    For every constant $c > 0$, there exists a constant $\alpha = \alpha(c) > 0$ for which the following statement is true: For any centrally symmetric convex body $K\seq[-1, 1]^m$ with $\vol_m(K)\geq c^m$, there exists $x\in \alpha K\cap [-1, 1]^m$ such that $\#\{i\in[m]:x_i = -1\}\geq m/4$. Furthermore such an $x$ can be found in randomized $\poly(m)$ time given a separation oracle for $K$. 
\end{theorem}
We can then replace $w$ with $w_X$ and iterate (with a bit of care for the errors). 

Thus, it suffices to show that if $|S| \gg (\log N)/\eps^2$, then $\vol(Q_\eps) \geq c^{|S|}$ for some universal constant $c > 0$. 

\paragraph{Bounding the volume of $Q_\eps$} This is the key part of the proof now. First, observe that we have two sets of constraints in $Q_\eps$: the \emph{upper} bound constraints and the lower bound constraints. We need to simultaneously satisfy both. Let us separate the two out:

\begin{align*}
Q^+_\eps &= \left\{X \in [-1,1]^S: \sum_{s \in S} X_s w(s) \lambda(\chi,s) \leq \eps \sum_{s \in S} w(s) \lambda(\chi,s),\;\;\forall \chi \in \widehat{G}\right\},\\
Q^-_\eps &= \left\{X \in [-1,1]^S: \sum_{s \in S} X_s w(s) \lambda(\chi,s) \geq - \eps \sum_{s \in S} w(s) \lambda(\chi,s),\;\;\forall \chi \in \widehat{G}\right\}.
\end{align*}

Clearly, $Q^-_\eps = - Q^+_\eps$ and we want to lower bound $\vol(Q_\eps) = \vol(Q^+_\eps \cap -Q^+_\eps)$. In general, getting a lower bound on the volume of $Q^+_\eps$, would not by itself imply a volume bound on the symmetrized form $Q^+_\eps \cap -Q^+_\eps$. However, \cite{ReisR26} introduced an approach where sufficiently strong volume lower bounds on all coordinate sections of a convex body $K$ implies a lower bound on the volume of $K \cap -K$: 
\begin{theorem}[Theorem 16 in \cite{ReisR26}]
\label{thm:large-symmetric-subset}
    Let $p, \eps\in(0, 1/2]$ be real numbers, and let $m\in\N$ be an integer such that $m\geq\log_2(1/p)/\eps^2$. Let $[-\eps, \eps]^m\seq K\seq[-1, 1]^m$ be a convex body such that for every non-empty $T\seq[m]$ we have $\vol_T(K_T)\geq p2^{|T|}$. Then $\vol_m(K\cap -K)\geq 2^{-5m}$. 
\end{theorem}

Thus, by using the above theorem, it now suffices to lower bound the volume of $Q^+_\eps$.\footnote{Note that $[-\eps, \eps]^S\seq Q^+_\eps$, and thus we meet that criterion in the hypothesis of \cref{thm:large-symmetric-subset}} We in fact show that the volume of even $Q^+_0$ is at least $2^{|S|}/N$. 

\begin{theorem}[Volume Estimate for the Asymmetric Sparsification Polytope]
\label{thm:volume-claim}
Let $S\subset G$ be an arbitrary set, and let $w:S\to[0, \infty)$ be an arbitrary map, and let $Q^+_0$ be as defined above. Then, $\vol_S(Q^+_0)\geq (1/N) \cdot 2^{|S|}$. 
\end{theorem}

The proof of this is by an elementary symmetry argument. For $X$ uniformly random over $[-1,1]^S$, by the symmetry of characters it follows that $$\Pr\left[\sup_{\chi\in\hat{G}} \sum_{s \in S} X_s w(s) \lambda(\chi,s) = \sum_{s \in S} X_s w(s) \lambda(\widehat{0}, s)\right] \geq \frac{1}{N}.$$

Combining the above bounds with suitably choosing an increasing schedule of $\eps_t$'s gives us \cref{thm:cayley-sparsifier}.

\section{Proof of \cref{thm:cayley-sparsifier}}

Fix a weighted Cayley graph $\Cay(G,S,w)$ as in the statement of \cref{thm:cayley-sparsifier}. We start by proving the volume lower bound on $Q^+_0$. To do so, we need some elementary properties of group-invariant random processes.

\begin{proposition}
\label{prop:invariant-process}
    Let $H$ be a finite abelian group, and let $(W_h)_{h\in H}$ be a $H$-invariant real-valued random process. Then 
    \[\Pr\left(W_0 = \inf_{h\in H}W_h\right) \geq \frac{1}{|H|}. \]
\end{proposition}
\begin{proof}
    Define the Boolean-valued random variable
    \[m_h:= \1\left(W_h = \inf_{h'\in H}W_{h'}\right).\]
    By definition we have $\sum_{h\in H}m_h\geq 1$. On taking expectations we obtain
    \[\sum_{h\in H}\Pr\left(W_h = \inf_{h'\in H}W_{h'}\right)\geq 1.\]
    By $H$-invariance we have, for all $h\in H$,
    \[\Pr\left(W_h = \inf_{h'\in H}W_{h'}\right) = \Pr\left(W_0 = \inf_{h'\in H}W_{h'}\right),\] and the result follows.
\end{proof}

\begin{lemma}
\label{lem:zero-dom}
    Let $\{\xi_s\}_{s\in S}$ be i.i.d. random variables uniformly sampled from $[-1, 1]$, and let $w:S\to[0, \infty)$ be some fixed map. For every $\chi\in\hat{G}$ define the real-valued random variable
    \[F(\chi):= \sum_{s\in S}\xi_sw(s)\Ree\chi(s).\]
    Then 
    \[\Pr\left(F(\hat{0}) =  \inf_{\chi\in\hat{G}}F(\chi)\right)\geq\frac{1}{N}. \]
\end{lemma}
\begin{proof}
    Let $\nu = (X, Y)$ be a rotationally invariant distribution supported on $\R^2$ with uniform marginals as in \cref{fact:rot-uniform}. Let $\{V_s\}_{s\in S}$ be i.i.d. random variables sampled from $\nu$, write $V_s = (X_s, Y_s)\in\R^2$, and also write $Z_s:= X_s + iY_s\in\C$. Note that $\{Z_s\}_{s\in S}$ are independent rotationally invariant complex-valued random variables.
    
    Now define the real-valued random variables
    \[\tilde{F}(\chi):= \sum_{s\in S}w(s)\cdot\Ree(Z_s\chi(s)) = \sum_{s\in S}w(s)\cdot\left(X_s\cdot\Ree\chi(s) - Y_s\cdot\Imm\chi(s)\right).\]
    We claim that $\tilde{F}$ is $\hat{G}$-invariant: Indeed, for any $\chi_0\in\hat{G}$ we have 
    \[\tilde{F}(\chi + \chi_0) = \sum_{s\in S}w(s)\cdot\Ree(Z_s\chi_0(s)\chi(s)) = \sum_{s\in S}\Ree(w(s)Z_s\chi_0(s)\chi(s)).\]
    Since $(w(s)Z_s)_{s\in S}$ are independent rotationally invariant random variables, the (joint) distribution of $(w(s)Z_s)_{s\in S}$ is the same as the distribution of $(w(s)Z_s\chi_0(s))_{s\in S}$ since $\chi_0(s)$ corresponds to a rotation. The $\hat{G}$-invariance of $\tilde{F}$ follows, and consequently, by \cref{prop:invariant-process} we have 
    \[\Pr\left(\tilde{F}(\hat{0}) =  \inf_{\chi\in\hat{G}}\tilde{F}(\chi)\right)\geq\frac{1}{N}.\]
    Now, note that for any $\chi\in\hat{G}$, we have 
    \[\tilde{F}(\hat{0}) - \tilde{F}(\pm\chi) = \sum_{s\in S}w(s)X_s(1 - \Ree\chi(s)) \pm \sum_{s\in S}w(s)Y_s\Imm\chi(s),\]
    where the $\pm$s correspond.\footnote{Note that $-\chi$ here represents the inverse of $\chi$ in $\hat{G}$, \textbf{not} the negation of $\chi$} Thus if $\tilde{F}(\hat{0}) =  \inf_{\chi\in\hat{G}}\tilde{F}(\chi)$ then for every $\chi\in\hat{G}$ we have
    \[\sum_{s\in S}w(s)X_s(1 - \Ree\chi(s))\leq -\left| \sum_{s\in S}w(s)Y_s\Imm\chi(s)\right|\leq 0\implies\sum_{s\in S}w(s)X_s\leq \sum_{s\in S}w(s)X_s\Ree\chi(s).\]
    Since this holds for every $\chi\in\hat{G}$, we obtain
    \[\tilde{F}(\hat{0}) =  \inf_{\chi\in\hat{G}}\tilde{F}(\chi)\implies\sum_{s\in S}w(s)X_s\leq \inf_{\chi\in\hat{G}}\sum_{s\in S}w(s)X_s\Ree\chi(s).\]
    Consequently, we have 
    \[\frac{1}{N}\leq \Pr\left(\tilde{F}(\hat{0}) =  \inf_{\chi\in\hat{G}}\tilde{F}(\chi)\right)\leq \Pr\left(\sum_{s\in S}w(s)X_s\leq \inf_{\chi\in\hat{G}}\sum_{s\in S}w(s)X_s\Ree\chi(s)\right)\]
    \[ = \Pr\left(\sum_{s\in S}w(s)\xi_s\leq \inf_{\chi\in\hat{G}}\sum_{s\in S}w(s)\xi_s\Ree\chi(s)\right) = \Pr\left(F(\hat{0})\leq \inf_{\chi\in\hat{G}}F(\chi)\right).\]
    since the distribution of $(X_s)_{s\in S}$ is the same as the distribution of $(\xi_s)_{s\in S}$, as desired.
\end{proof}

We can now prove our volume estimate for the asymmetric sparsification polytope $Q^+_0$, \cref{thm:volume-claim}.

\begin{proof}[Proof of \cref{thm:volume-claim}]
Note that 
\[Q^+_0 = \left\{X\in[-1, 1]^S:\sum_{s\in S}w(s)X_s\leq \inf_{\chi\in\hat{G}}\sum_{s\in S}w(s)X_s\Ree\chi(s)\right\}\]
by the definition of $\lambda(\chi, s)$. Also note that 
\[\frac{\vol_S(Q^+_0)}{2^{|S|}} = \Pr_{X\sim[-1, 1]^S}\left(\sum_{s\in S}w(s)X_s\leq \inf_{\chi\in\hat{G}}\sum_{s\in S}w(s)X_s\Ree\chi(s)\right).\]
Thus for the random variable $X$ sampled uniformly from $[-1, 1]^S$, define for any $\chi\in\hat{G}$,
\[F(\chi):= \sum_{s\in S}w(s)X_s\Ree\chi(s),\]
and note that 
\[\Pr_{X\sim[-1, 1]^S}\left(\sum_{s\in S}w(s)X_s\leq \inf_{\chi\in\hat{G}}\sum_{s\in S}w(s)X_s\Ree\chi(s)\right) = \Pr_{X\sim[-1, 1]^S}\left(F(\hat{0})\leq\inf_{\chi\in\hat{G}}F(\chi)\right),\]
and we are now done by \cref{lem:zero-dom}.   
\end{proof}

We now prove our main theorem---\cref{thm:cayley-sparsifier}.

\begin{proof}[Proof of \cref{thm:cayley-sparsifier}]

As described in the introduction, the basic idea is to iteratively prune the support of $S$. We start with the weight function $w^{(0)} = w\in\R_{\geq 0}^S$.\footnote{Here we view functions $S\to\R$ as vectors in $\R^S$}

We iteratively construct a series of vectors $\{w^{(t)}\}_{0\leq t\leq T}\in\R_{\geq 0}^S$ such that $S_t\subset S_{t - 1}$ and $m_t\leq 0.75m_{t - 1}$ for all $t\geq 1$, where $S_t:= \supp(w^{(t)})$ and $m_t:= |S_t|$. The stopping time $T$ is chosen to be the smallest integer for which $m_T\leq C_0\eps^{-2}\log(N)$ for some large enough absolute constant $C_0 > 0$. Note that $T\leq O(\log N)$. Also write $\eps_t:= \sqrt{m_t^{-1}\cdot\log_2(N)}$.

Now for any $0\leq t < T$ consider the sparsification polytope
\begin{equation}\label{eq:qepst-def}
\begin{aligned}
Q_{\eps_t}:=\Bigl\{X\in[-1,1]^{S_t}:\;&
(1-\eps_t)\sum_{s\in S_t}w^{(t)}(s)\lambda(\chi,s)
\le
\sum_{s\in S_t}(1+X_s)w^{(t)}(s)\lambda(\chi,s)\\
&
\le
(1+\eps_t)\sum_{s\in S_t}w^{(t)}(s)\lambda(\chi,s)
\quad\forall\,\chi\in\hat{G}
\Bigr\}.
\end{aligned}
\end{equation}
Note that since $Q_{\eps_t} = Q_{\eps_t}^+ \cap -Q_{\eps_t}^+$, $Q_{\eps_t}$ is a centrally symmetric convex body. Furthermore, note that $\vol_T((Q_{\eps_t}^+)_T)\geq \vol_T((Q_{0}^+)_T)\overset{\text{\cref{thm:volume-claim}}}{\geq} (1/N)\cdot 2^{|T|}$ for any non-empty $T\seq S_t$.\footnote{Note that $(Q_{0}^+)_T$ is exactly the convex body you get when you replace $S_t$ in \cref{eq:qepst-def} with $T$, and thus \cref{thm:volume-claim} applies} Consequently, by \cref{thm:large-symmetric-subset} we have $\vol(Q_{\eps_t})\geq 2^{-5|S_t|}$. Hence by \cref{thm:good-partial-coloring} we obtain that there exists $x\in \alpha Q_{\eps_t}\cap [-1, 1]^{S_t}$ such that $\#\{s\in S_t:x_s = -1\}\geq m_t/4$ for some absolute constant $\alpha > 0$. Since a separation oracle for $Q_{\eps_t}$ can be implemented in (deterministic) $\poly(m, N) \leq \poly(N)$ time, we can compute $x$ in randomized $\poly(m)\cdot\poly(N) \leq \poly(N)$ time.

Write this $x$ as $x^{(t)}$ and update $w^{(t + 1)}:= w^{(t)}\odot(\1 + x^{(t)})$ (filling $w^{(t + 1)}$ with zeros in $S\setminus S_t$). Notice that all the desired invariants (such as $S_{t + 1}\subset S_t$ and $m_{t + 1}\leq 0.75m_t$) are maintained by this update. Moreover, by the definition of $Q_{\eps_t}$ we have 
\[(1-\alpha\eps_t)\sum_{s\in S_t}w^{(t)}(s)\lambda(\chi,s)
\le
\sum_{s\in S_t}(1+X_s)w^{(t)}(s)\lambda(\chi,s)\le
(1+\alpha\eps_t)\sum_{s\in S_t}w^{(t)}(s)\lambda(\chi,s)
\quad\forall\,\chi\in\hat{G}\]
Consequently when the process ends, we have 
\[\prod_{t = 0}^{T - 1}(1-\alpha\eps_t)\cdot\sum_{s\in S_0}w(s)\lambda(\chi,s)
\le
\sum_{s\in S_T}w^{(T)}(s)\lambda(\chi,s)\le
\prod_{t = 0}^{T - 1}(1+\alpha\eps_t)\cdot\sum_{s\in S_0}w(s)\lambda(\chi,s)
\quad\forall\,\chi\in\hat{G}\]
Note that since $\eps_t\leq C_0^{-1/2}\eps\cdot(3/4)^{(T - t - 1)/2}$, we have
\[\prod_{t = 0}^{T - 1}(1 - \alpha\eps_t)\geq 1 - \sum_{t = 0}^{T - 1}\alpha\eps_t\geq 1 - \alpha\eps C_0^{-1/2}\sum_{t = 0}^{T - 1}\left(\frac{3}{4}\right)^{(T - t - 1)/2}\geq 1 - O(\alpha\eps C_0^{-1/2}).\]
Similarly, 
\[\prod_{t = 0}^{T - 1}(1 + \alpha\eps_t)\leq \exp\left(\sum_{t = 0}^{T - 1}\alpha\eps_t\right)\leq\exp\left(O(\alpha\eps C_0^{-1/2})\right).\]
Consequently, if $\alpha C_0^{-1/2}\ll 1/2$, i.e. $C_0\gg \alpha^2$, then we have 
\[1 - \eps\leq\prod_{t = 0}^{T - 1}(1 - \alpha\eps_t)\leq\prod_{t = 0}^{T - 1}(1 + \alpha\eps_t)\leq 1 + \eps,\]
as desired. Here we use the fact that $e^x\leq 1 + 2x$ for all $x\in[0, 1]$.

Finally, to meet symmetry requirements, notice that for any weight function $w:G\to[0, \infty)$ (where we set $w(g) = 0$ if $g\notin\supp(w)$), if we define $\tilde{w}:G\to[0, \infty)$ as $\tilde{w}(g):= \frac{w(g) + w(-g)}{2}$, then $\tilde{w}$ is symmetric, $|\supp(\tilde{w})|\leq 2|\supp(w)|$, and for any $\chi\in\hat{G}$ we have $\sum_{s\in S}w(s)\lambda(\chi, s) = \sum_{s\in S}\tilde{w}(s)\lambda(\chi, s)$. Consequently, at the cost of blowing up the support size of $w^{(T)}$ by a factor of $2$, we obtain a symmetric function (which can be computed in randomized $\poly(N)$ time) $w':S\to[0, \infty)$ \footnote{notice that since $S$ is symmetric the ``symmetrization'' process maintains that $\supp(w')\seq S$}. By \cref{fact:cayley-eig} $w'$ meets the stated requirements of the theorem, as desired.
\end{proof}
\section{Acknowledgments}
\parhead{AI Acknowledgment} After \cite{ReisR26} was posted online, the authors realized that ideas from the paper could be useful in sparsifying abelian Cayley graphs. The authors isolated \cref{thm:volume-claim} as a suitable analog of \cite[Theorem~15]{ReisR26} that would imply \cref{thm:cayley-sparsifier} as a consequence. The simple proof of \cref{thm:volume-claim} was found by multiple sequential sessions of ChatGPT-5.5-Plus. The authors wrote the proof of this claim with suitable modifications for clarity. 

A.B. thanks Louie Putterman and Josh Brakensiek for useful discussions. We also thank anonymous reviewers for pointing out typos and for helpful comments.

\bibliographystyle{alphaurl}
\bibliography{abelian}

@misc{ReisR26,
      title={Linear-size $\ell_1$ sparsifiers}, 
      author={Victor Reis and Thomas Rothvoss},
      year={2026},
      eprint={2606.28147},
      archivePrefix={arXiv},
      primaryClass={math.MG},
      url={https://arxiv.org/abs/2606.28147v1}, 
}

@inproceedings{BenczurK96,
author = {Bencz\'{u}r, Andr\'{a}s A. and Karger, David R.},
title = {Approximating s-t minimum cuts in \~{O}(n2) time},
year = {1996},
isbn = {0897917855},
publisher = {Association for Computing Machinery},
address = {New York, NY, USA},
url = {https://doi.org/10.1145/237814.237827},
doi = {10.1145/237814.237827},
booktitle = {Proceedings of the Twenty-Eighth Annual ACM Symposium on Theory of Computing},
pages = {47–55},
numpages = {9},
location = {Philadelphia, Pennsylvania, USA},
series = {STOC '96}
}

@article{SpielmanT11,
author = {Spielman, Daniel A. and Teng, Shang-Hua},
title = {Spectral Sparsification of Graphs},
journal = {SIAM Journal on Computing},
volume = {40},
number = {4},
pages = {981-1025},
year = {2011},
doi = {10.1137/08074489X},
URL = {https://doi.org/10.1137/08074489X},
}

@article{SpielmanS11,
author = {Spielman, Daniel A. and Srivastava, Nikhil},
title = {Graph Sparsification by Effective Resistances},
journal = {SIAM Journal on Computing},
volume = {40},
number = {6},
pages = {1913-1926},
year = {2011},
doi = {10.1137/080734029},
}

@article{BatsonSS14,
author = {Batson, Joshua and Spielman, Daniel A. and Srivastava, Nikhil},
title = {Twice-Ramanujan Sparsifiers},
journal = {SIAM Review},
volume = {56},
number = {2},
pages = {315-334},
year = {2014},
doi = {10.1137/130949117},
URL = {https://doi.org/10.1137/130949117},
}

@inproceedings{KhannaPS24,
author={Khanna, Sanjeev and Putterman, Aaron and Sudan, Madhu},
  editor       = {David P. Woodruff},
  title        = {Code sparsification and its applications},
  booktitle    = {Proceedings of the 2024 {ACM-SIAM} Symposium on Discrete Algorithms,
                  {SODA} 2024, Alexandria, VA, USA, January 7-10, 2024},
  pages        = {5145--5168},
  publisher    = {{SIAM}},
  year         = {2024},
  url          = {https://doi.org/10.1137/1.9781611977912.185},
  doi          = {10.1137/1.9781611977912.185},
  bibsource    = {dblp computer science bibliography, https://dblp.org}
}

@inproceedings{KhannaPS25,
author = {Khanna, Sanjeev and Putterman, Aaron and Sudan, Madhu},
title = {Efficient Algorithms and New Characterizations for CSP Sparsification},
year = {2025},
isbn = {9798400715105},
publisher = {Association for Computing Machinery},
address = {New York, NY, USA},
url = {https://doi.org/10.1145/3717823.3718205},
doi = {10.1145/3717823.3718205},
booktitle = {Proceedings of the 57th Annual ACM Symposium on Theory of Computing},
pages = {407–416},
numpages = {10},
location = {Prague, Czechia},
series = {STOC '25}
}

@inproceedings{BrakensiekG25,
author = {Brakensiek, Joshua and Guruswami, Venkatesan},
title = {Redundancy Is All You Need},
year = {2025},
isbn = {9798400715105},
publisher = {Association for Computing Machinery},
address = {New York, NY, USA},
url = {https://doi.org/10.1145/3717823.3718212},
doi = {10.1145/3717823.3718212},
booktitle = {Proceedings of the 57th Annual ACM Symposium on Theory of Computing},
pages = {1614–1625},
numpages = {12},
location = {Prague, Czechia},
series = {STOC '25}
}

@inproceedings{HsiehLMPZ26,
  author       = {Jun{-}Ting Hsieh and
                  Daniel Z. Lee and
                  Sidhanth Mohanty and
                  Aaron Putterman and
                  Rachel Yun Zhang},
  editor       = {Kasper Green Larsen and
                  Barna Saha},
  title        = {Sparsifying Cayley Graphs on Every Group},
  booktitle    = {Proceedings of the 2026 Annual {ACM-SIAM} Symposium on Discrete Algorithms,
                  {SODA} 2026, Vancouver, BC, Canada, January 11-14, 2026},
  pages        = {6029--6041},
  publisher    = {{SIAM}},
  year         = {2026},
  url          = {https://doi.org/10.1137/1.9781611978971.215},
  doi          = {10.1137/1.9781611978971.215},
  bibsource    = {dblp computer science bibliography, https://dblp.org}
}

@inproceedings{Lee23,
  author       = {James R. Lee},
  editor       = {Barna Saha and
                  Rocco A. Servedio},
  title        = {Spectral Hypergraph Sparsification via Chaining},
  booktitle    = {Proceedings of the 55th Annual {ACM} Symposium on Theory of Computing,
                  {STOC} 2023, Orlando, FL, USA, June 20-23, 2023},
  pages        = {207--218},
  publisher    = {{ACM}},
  year         = {2023},
  url          = {https://doi.org/10.1145/3564246.3585165},
  doi          = {10.1145/3564246.3585165},
  bibsource    = {dblp computer science bibliography, https://dblp.org}
}

@inproceedings{JambulapatiLS23,
  author       = {Arun Jambulapati and
                  Yang P. Liu and
                  Aaron Sidford},
  editor       = {Barna Saha and
                  Rocco A. Servedio},
  title        = {Chaining, Group Leverage Score Overestimates, and Fast Spectral Hypergraph
                  Sparsification},
  booktitle    = {Proceedings of the 55th Annual {ACM} Symposium on Theory of Computing,
                  {STOC} 2023, Orlando, FL, USA, June 20-23, 2023},
  pages        = {196--206},
  publisher    = {{ACM}},
  year         = {2023},
  url          = {https://doi.org/10.1145/3564246.3585136},
  doi          = {10.1145/3564246.3585136},
  bibsource    = {dblp computer science bibliography, https://dblp.org}
}

@inbook{BasuKLM26,
author = {Arpon Basu and Pravesh K. Kothari and Yang P. Liu and Raghu Meka},
title = {Sparsifying Sums of Positive Semidefinite Matrices},
booktitle = {Proceedings of the 2026 Annual ACM-SIAM Symposium on Discrete Algorithms (SODA)},
pages = {6042-6064},
doi = {10.1137/1.9781611978971.216},
URL = {https://doi.org/10.1137/1.9781611978971.216},
year = {2026}
}

@article{FiltserK17,
  author       = {Arnold Filtser and
                  Robert Krauthgamer},
  title        = {Sparsification of Two-Variable Valued Constraint Satisfaction Problems},
  journal      = {{SIAM} J. Discret. Math.},
  volume       = {31},
  number       = {2},
  pages        = {1263--1276},
  year         = {2017},
  url          = {https://doi.org/10.1137/15M1046186},
  doi          = {10.1137/15M1046186},
  bibsource    = {dblp computer science bibliography, https://dblp.org}
}

@inproceedings{AbrahamDKKP16,
  author       = {Ittai Abraham and
                  David Durfee and
                  Ioannis Koutis and
                  Sebastian Krinninger and
                  Richard Peng},
  editor       = {Irit Dinur},
  title        = {On Fully Dynamic Graph Sparsifiers},
  booktitle    = {{IEEE} 57th Annual Symposium on Foundations of Computer Science, {FOCS}
                  2016, Hyatt Regency, New Brunswick, New Jersey, USA, October 9-11,
                  2016},
  pages        = {335--344},
  publisher    = {{IEEE} Computer Society},
  year         = {2016},
  url          = {https://doi.org/10.1109/FOCS.2016.44},
  doi          = {10.1109/FOCS.2016.44},
  bibsource    = {dblp computer science bibliography, https://dblp.org}
}

@inproceedings{KoganK15,
  author       = {Dmitry Kogan and
                  Robert Krauthgamer},
  editor       = {Tim Roughgarden},
  title        = {Sketching Cuts in Graphs and Hypergraphs},
  booktitle    = {Proceedings of the 2015 Conference on Innovations in Theoretical Computer
                  Science, {ITCS} 2015, Rehovot, Israel, January 11-13, 2015},
  pages        = {367--376},
  publisher    = {{ACM}},
  year         = {2015},
  url          = {https://doi.org/10.1145/2688073.2688093},
  doi          = {10.1145/2688073.2688093},
  bibsource    = {dblp computer science bibliography, https://dblp.org}
}

@inproceedings{KapralovKTY21,
  author       = {Michael Kapralov and
                  Robert Krauthgamer and
                  Jakab Tardos and
                  Yuichi Yoshida},
  title        = {Spectral Hypergraph Sparsifiers of Nearly Linear Size},
  booktitle    = {62nd {IEEE} Annual Symposium on Foundations of Computer Science, {FOCS}
                  2021, Denver, CO, USA, February 7-10, 2022},
  pages        = {1159--1170},
  publisher    = {{IEEE}},
  year         = {2021},
  url          = {https://doi.org/10.1109/FOCS52979.2021.00114},
  doi          = {10.1109/FOCS52979.2021.00114},
  bibsource    = {dblp computer science bibliography, https://dblp.org}
}

@inproceedings{KapralovKTY21a,
  author       = {Michael Kapralov and
                  Robert Krauthgamer and
                  Jakab Tardos and
                  Yuichi Yoshida},
  editor       = {Samir Khuller and
                  Virginia Vassilevska Williams},
  title        = {Towards tight bounds for spectral sparsification of hypergraphs},
  booktitle    = {{STOC} '21: 53rd Annual {ACM} {SIGACT} Symposium on Theory of Computing,
                  Virtual Event, Italy, June 21-25, 2021},
  pages        = {598--611},
  publisher    = {{ACM}},
  year         = {2021},
  url          = {https://doi.org/10.1145/3406325.3451061},
  doi          = {10.1145/3406325.3451061},
  bibsource    = {dblp computer science bibliography, https://dblp.org}
}

@inproceedings{ChenKN20,
  author       = {Yu Chen and
                  Sanjeev Khanna and
                  Ansh Nagda},
  editor       = {Sandy Irani},
  title        = {Near-linear Size Hypergraph Cut Sparsifiers},
  booktitle    = {61st {IEEE} Annual Symposium on Foundations of Computer Science, {FOCS}
                  2020, Durham, NC, USA, November 16-19, 2020},
  pages        = {61--72},
  publisher    = {{IEEE}},
  year         = {2020},
  url          = {https://doi.org/10.1109/FOCS46700.2020.00015},
  doi          = {10.1109/FOCS46700.2020.00015},
  bibsource    = {dblp computer science bibliography, https://dblp.org}
}

@inproceedings{AhnGM12b,
  author       = {Kook Jin Ahn and
                  Sudipto Guha and
                  Andrew McGregor},
  editor       = {Michael Benedikt and
                  Markus Kr{\"{o}}tzsch and
                  Maurizio Lenzerini},
  title        = {Graph sketches: sparsification, spanners, and subgraphs},
  booktitle    = {Proceedings of the 31st {ACM} {SIGMOD-SIGACT-SIGART} Symposium on
                  Principles of Database Systems, {PODS} 2012, Scottsdale, AZ, USA,
                  May 20-24, 2012},
  pages        = {5--14},
  publisher    = {{ACM}},
  year         = {2012},
  url          = {https://doi.org/10.1145/2213556.2213560},
  doi          = {10.1145/2213556.2213560},
  bibsource    = {dblp computer science bibliography, https://dblp.org}
}

@article{ButtiZ20,
  title = {Sparsification of {{Binary CSPs}}},
  author = {Butti, Silvia and {\v Z}ivn{\'y}, Stanislav},
  year = {2020},
  month = jan,
  journal = {SIAM Journal on Discrete Mathematics},
  volume = {34},
  number = {1},
  pages = {825--842},
  publisher = {{Society for Industrial and Applied Mathematics}},
  issn = {0895-4801},
  doi = {10.1137/19M1242446}
}

@inproceedings{SpielmanT04,
author = {Spielman, Daniel A. and Teng, Shang-Hua},
title = {Nearly-linear time algorithms for graph partitioning, graph sparsification, and solving linear systems},
year = {2004},
isbn = {1581138520},
publisher = {Association for Computing Machinery},
address = {New York, NY, USA},
url = {https://doi.org/10.1145/1007352.1007372},
doi = {10.1145/1007352.1007372},
booktitle = {Proceedings of the Thirty-Sixth Annual ACM Symposium on Theory of Computing},
pages = {81–90},
numpages = {10},
location = {Chicago, IL, USA},
series = {STOC '04}
}

@article{KapralovLMMS17,
  title={Single pass spectral sparsification in dynamic streams},
  author={Kapralov, Michael and Lee, Yin Tat and Musco, CN and Musco, Christopher Paul and Sidford, Aaron},
  journal={SIAM Journal on Computing},
  volume={46},
  number={1},
  pages={456--477},
  year={2017},
  publisher={SIAM}
}

@article{McGregor14,
  title={Graph stream algorithms: a survey},
  author={McGregor, Andrew},
  journal={ACM SIGMOD Record},
  volume={43},
  number={1},
  pages={9--20},
  year={2014},
  publisher={ACM New York, NY, USA}
}

@article{AlonR94,
  title={Random Cayley graphs and expanders},
  author={Alon, Noga and Roichman, Yuval},
  journal={Random Structures \& Algorithms},
  volume={5},
  number={2},
  pages={271--284},
  year={1994},
  publisher={Wiley Online Library}
}

@InProceedings{JalanM21,
  author =	{Jalan, Akhil and Moshkovitz, Dana},
  title =	{{Near-Optimal Cayley Expanders for Abelian Groups}},
  booktitle =	{41st IARCS Annual Conference on Foundations of Software Technology and Theoretical Computer Science (FSTTCS 2021)},
  pages =	{24:1--24:23},
  series =	{Leibniz International Proceedings in Informatics (LIPIcs)},
  ISBN =	{978-3-95977-215-0},
  ISSN =	{1868-8969},
  year =	{2021},
  volume =	{213},
  editor =	{Boja\'{n}czyk, Miko{\l}aj and Chekuri, Chandra},
  publisher =	{Schloss Dagstuhl -- Leibniz-Zentrum f{\"u}r Informatik},
  address =	{Dagstuhl, Germany},
  URL =		{https://drops.dagstuhl.de/entities/document/10.4230/LIPIcs.FSTTCS.2021.24},
  URN =		{urn:nbn:de:0030-drops-155359},
  doi =		{10.4230/LIPIcs.FSTTCS.2021.24}
}

@article{Rothvoss17,
  title={Constructive discrepancy minimization for convex sets},
  author={Rothvoss, Thomas},
  journal={SIAM Journal on Computing},
  volume={46},
  number={1},
  pages={224--234},
  year={2017},
  publisher={SIAM}
}

\end{document}